\documentclass[journal,twoside,web]{ieeecolor} 

\usepackage{lcsys}
\usepackage{algorithmic}
\usepackage{textcomp}
\def\BibTeX{{\rm B\kern-.05em{\sc i\kern-.025em b}\kern-.08em
    T\kern-.1667em\lower.7ex\hbox{E}\kern-.125emX}}
\makeatletter
\let\NAT@parse\undefined
\makeatother

\usepackage{amsmath}
\usepackage{amssymb,bm}
\usepackage{amsfonts}
\usepackage{enumerate}
\usepackage{tabulary}
\usepackage{multirow}
\usepackage{cite}
\usepackage{lipsum} 
\usepackage{verbatim}
\usepackage[hidelinks]{hyperref}
\usepackage{url}
\usepackage{balance}
\usepackage{siunitx}
\usepackage{ifthen}
\usepackage{forloop}
\usepackage{listings}
\usepackage{booktabs}
\usepackage[lined,boxed,commentsnumbered,linesnumbered]{algorithm2e}
\usepackage{cleveref}

\usepackage[utf8]{inputenc}
\usepackage{pgfplots}
\DeclareUnicodeCharacter{2212}{−}
\usepgfplotslibrary{groupplots,dateplot}
\usetikzlibrary{patterns,shapes.arrows}
\pgfplotsset{compat=newest}
\pgfplotsset{every tick label/.append style={font=\scriptsize}}

\hypersetup{
  colorlinks    = true, 
  urlcolor      = blue, 
  linkcolor     = blue, 
  citecolor     = blue   
}
\definecolor{codegreen}{rgb}{0,0.6,0}
\definecolor{codepurple}{rgb}{0.58,0,0.82}
\definecolor{backcolour}{rgb}{0.95,0.95,0.92}
\lstdefinestyle{buzz}{
    backgroundcolor=\color{black!5},   
    commentstyle=\color{codegreen},
    keywordstyle=\color{blue},
    numberstyle=\tiny\color{black!30},
    stringstyle=\color{codepurple},
    basicstyle=\footnotesize\ttfamily,
    breakatwhitespace=false,         
    breaklines=true,                 
    captionpos=b,                    
    keepspaces=true,                 
    numbers=left,                    
    numbersep=5pt,                  
    showspaces=false,                
    showstringspaces=false,
    showtabs=false,                  
    tabsize=2,
}
\renewcommand{\vec}[1]{#1}
\newcommand{\xdot}{\dot{x}}
\newcommand{\x}{x}
\renewcommand{\u}{u}
\newcommand{\f}{f}
\newcommand{\g}{g}

\renewcommand{\u}{\vec{u}}

\newcommand{\set}[1]{\mathbb{#1}}
\newcommand{\R}{\mathbb{R}}

\newcommand{\K}{\mathcal{K}}

\usepackage{stackengine}
\newcommand{\gammalb}{\underline{\gamma}}

\newtheorem{theorem}{Theorem}
\newtheorem{corollary}{Corollary}

\newtheorem{definition}{Definition}

\newtheorem{remark}{Remark}
\newtheorem{assumption}{Assumption}

\newcommand{\todo}[1]{\textcolor{black}{#1}}

\SetAlCapSkip{-0.25em}
\title{\LARGE \bf
Optimized Control Invariance Conditions for Uncertain Input-Constrained Nonlinear Control Systems
}

\author{Lukas Brunke, Siqi Zhou, Mingxuan Che, and Angela P. Schoellig
\thanks{The authors are with the 
\href{http://www.learnsyslab.org}{Learning Systems and Robotics Lab} at the Technical University of Munich, Germany and 
the Munich Institute of Robotics and Machine Intelligence~(MIRMI). 
LB and APS are also affiliated with the 
	University of Toronto Institute for Aerospace Studies, the  University of Toronto Robotics Institute, and the Vector Institute for Artificial Intelligence. Emails:
	\{lukas.brunke, siqi.zhou, mingxuan.che, angela.schoellig\}@tum.de}%
}

\begin{document}
\maketitle
\thispagestyle{empty}
\pagestyle{empty}


\begin{abstract}
Providing safety guarantees for learning-based controllers is important for real-world applications. One approach to realizing safety for arbitrary control policies is safety filtering. If necessary, the filter modifies control inputs to ensure that the trajectories of a closed-loop system stay within a given state constraint set for all future time, referred to as the set being positive invariant or the system being safe. Under the assumption of fully known dynamics, safety can be certified using control barrier functions (CBFs). However, the dynamics model is often either unknown or only partially known in practice. Learning-based methods have been proposed to approximate the CBF condition for unknown or uncertain systems from data; however, these techniques do not account for input constraints and, as a result, may not yield a valid CBF condition to render the safe set invariant. In this work, we study conditions that guarantee control invariance of the system under input constraints and propose an optimization problem to reduce the conservativeness of CBF-based safety filters. Building on these theoretical insights, we further develop a probabilistic learning approach that allows us to build a safety filter that guarantees safety for uncertain, input-constrained systems with high probability. We demonstrate the efficacy of our proposed approach in simulation and real-world experiments on a quadrotor and show that we can achieve safe closed-loop behavior for a learned system while satisfying state and input constraints.
\end{abstract}


\section{Introduction}
\label{sec:introduction}

\IEEEPARstart{I}{n} recent years, learning-based control techniques have enabled robots to perform complex tasks in uncertain environments. 
However, providing rigorous safety guarantees of learning-based controllers is typically challenging in such settings, which hinders their real-world deployment~\cite{DSL2021}. 
One approach to providing safety guarantees for learning-based controllers is to use a safety filter, which modifies the input to the system when it is deemed unsafe. 

Control barrier function~(CBF)-based safety filters~\cite{Ames2014} are an example of safety filters for continuous-time systems. A CBF-based safety filter requires a CBF which defines the safe set.
The synthesis of a CBF for nonlinear systems typically necessitates considerable offline computation~\cite{Dai2022}. 
However, once the CBF has been computed and under the assumption of control affine dynamics, safety can be efficiently determined online by solving a quadratic program~(QP)~\cite{ames2019a}, for which adequate solvers exist.
Motivated by the goal of achieving safe real-world operation of resource-limited hardware systems, for example, miniature quadrotors, we focus on CBF-based safety filters in this paper due to their reduced online computation. 

\begin{figure}[t]
    \centering
    \includegraphics[width=\columnwidth, height=2.8cm]{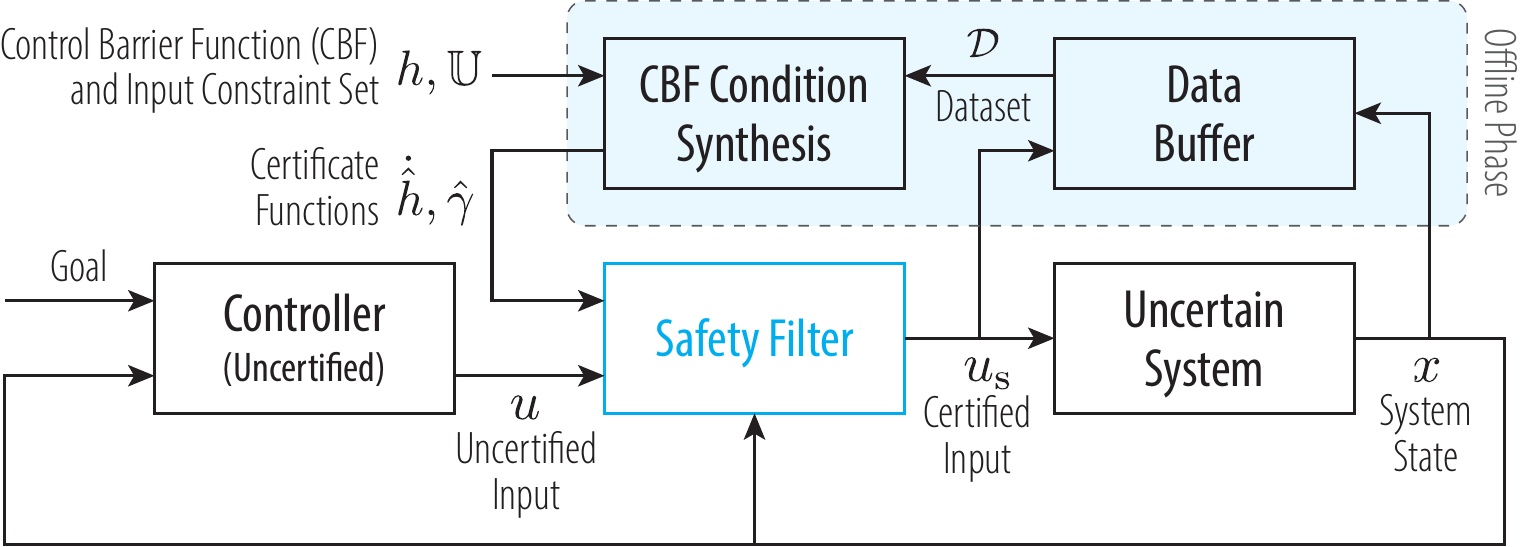}
    \vspace*{-1em}
    \caption{A block diagram of our proposed control barrier function safety filter approach for uncertain, input-constrained systems. 
    Using the theoretical conditions we derived, we can verify that our proposed safety filter satisfies safety under input constraints.  
    We can augment any unsafe controller with our proposed safety filter to guarantee safety with high probability. }
    \label{fig:blockdiagram}
\end{figure}

CBFs provide a scalar condition to determine whether a proposed control input keeps the system in the safe set~\cite{ames2019a}.
The CBF condition is derived from Nagumo's theorem~\cite{Nagumo1942berDL}, which intuitively defines safety by guaranteeing that 
for any state on the safe set's boundary, the system will not leave the safe set. 
CBFs can be synthesized for systems with fully known dynamics ~\cite{mitchell2005time, wang2023a, Dai2022, Wei2023} or
for unknown dynamics~\cite{Qin2022}. 

Real-world systems typically have actuation limits. For safety, these input constraints must be accounted for. Otherwise, the assumption of additional control authority can lead to safety violations. For systems with fully known dynamics, the authors in~\cite{xiao2022, Zeng2021} propose adapting the CBF condition's lower bound online to achieve feasibility under control input constraints.
If control invariance cannot be satisfied due to input constraints, the authors in~\cite{Agrawal2021} propose a method for determining control invariant subsets of the original safe set. 
In~\cite{Wei2023}, the authors study robust control invariance under input constraints for uncertain systems. 

In recent literature, learning-based CBF safety filters have been proposed to guarantee safety for uncertain systems~\cite{DSL2021}.  
Given a CBF, these approaches either directly learn the Lie derivative from data~\cite{taylor2020a, l4dc22} or indirectly learn it by first identifying a dynamics model of the uncertain system~\cite{Wang2018a, Ohnishi2019}.
In~\cite{Wang2018a, l4dc22, Castaneda2021}, probabilistic learning techniques for CBFs have been proposed to 
guarantee the system's safety with high probability. While effective, these works do not account for control input constraints, which may not lead to valid CBF conditions to render the system safe on the original state constraint set. 
Recently, the authors in~\cite{Nejati2023} have considered a data-based approach to synthesize a CBF-based safety controller for an \todo{unknown system with a discrete input set}.
In our work, we learn the system's Lie derivative that satisfies the CBF condition under general compact input constraints with high probability for an uncertain system, \todo{where the dynamics are (partially) unknown}.

Our contributions are as follows:
\textit{(i)}~derivation of  control invariance conditions for generic control-affine nonlinear systems under input constraints, 
\textit{(ii)} optimization-based design of CBF conditions to increase the set of feasible inputs while
guaranteeing control invariance under input constraints, 
\textit{(iii)} a probabilistic learning approach to obtain a CBF safety filter that provides probabilistic safety guarantees for uncertain, input-constrained systems, and \textit{(iv)}  simulation and real-world results demonstrating the efficacy of our approach. 


\section{Problem Formulation}
\label{sec:formulation}
In this work, we consider the control architecture in~\autoref{fig:blockdiagram} and a continuous-time nonlinear control system represented in the following control-affine form:
\begin{equation}
\label{eq:nonlinear_affine_control}
	\xdot = \f(\x) + \g (\x) \:\u\, ,
\end{equation}
where $\x\in \set{X} \subset \R^n$ is the state of the system with $\set{X}$ being the set of admissible states, $\u\in \set{U} \subset \R^m$ is the input of the system with $\set{U}$ being the set of admissible inputs, and  $\f:\R^n\mapsto \R^n$ and $\g:\R^n\mapsto \R^{n\times m}$ are locally Lipschitz continuous functions. We assume that $\set{X}$ and $\set{U}$ are known compact sets, and $f$ and $g$ are (partially) unknown functions.

The safe set $\set{C}\subseteq \set{X}$  is assumed to be given and is defined as the zero-superlevel set of a continuously differentiable function $h: \R^n \to \R$: $\set{C} = \{\x \in \set{X} \: \vert \: h(\x) \geq 0 \} $
where the boundary of the safe set is $\partial\set{C} = \{\x \in \set{X} \: \vert \: h(\x) = 0 \}$ with $\partial h(\x) / \partial \x  \neq 0$ for all $\x \in \partial \set{C}$, and the interior is $\text{Int}(\set{C}) = \{ x \in \set{X} \:\vert\: h(x) > 0\}$. The corresponding unsafe set is defined as $\bar{\set{C}}=\{ x \in \set{X}\:\vert\: h(x) < 0\}$.
Our goal is to augment a given, potentially unsafe state-feedback controller~$\pi(x)$ with a safety filter such that the uncertain system is safe (i.e., the system's state~$\x$ stays inside a safe set $\set{C}$ if it starts inside of $\set{C}$) under the input constraints $\set{U}$. 

\section{Background}
\label{sec:background}

To facilitate our discussion, we introduce relevant definitions and background on CBF safety filters in this section. 
\begin{definition}[Extended class-$\K$ function~\cite{ames2019a}]
\label{def:kappa_inf_extended}
    A function $\gamma: \R \to \R$ is said to be of class-$\K_e$ 
    if it is continuous, $\gamma(0) = 0$, and strictly increasing. 
\end{definition}

\begin{definition}[Positively control invariant set] Let $\mathfrak{U}$ be the set of all bounded controllers $\nu : \R_{\geq 0} \to \set{U}\,$. A set ${\set{C}\subseteq\set{X}}$ is a positively control invariant set for the control system in~\eqref{eq:nonlinear_affine_control} if~$\: \forall \: {\x_0 \in\set{C}} \,,\: \exists \: \nu \in \mathfrak{U} \,, \: \forall \: t \in \set{T}_{\x_0}^+ \,,\: \phi(t, \x_0, \nu) \in \set{C}$, where $\phi(t, \x_0, \nu)$ is the system's phase flow starting at $\x_0$ under the controller $\nu$, and $\set{T}_{\x_0}^+$ is the maximum time interval.
\end{definition}

\begin{definition}[CBF~\cite{ames2019a}]
	Let $\set{C} \subseteq \set{X}$ be the superlevel set of a continuously differentiable function $h: \set{X} \to \R$, then $h$ is a CBF if there exists a class-$\K_e$ function $\gamma$ such that for all $\x \in \set{X}$ the control system in~\eqref{eq:nonlinear_affine_control} satisfies
	\begin{equation}
		\label{eq:cbf_lie_derivative}
		\max_{\u \in \set{U}} \left[L_\vec{f} h(\x) + L_\vec{g} h(\x) \u \right] \geq - \gamma(h(x)) \,,
	\end{equation}
\end{definition}
where $L_f h(x) $ and $L_g h(x)$ are the Lie derivatives of $h$ along $f$ and $g$, respectively. In the following, we write $\dot{h}(x, u) = L_\vec{f} h(\x) + L_\vec{g} h(\x) \u$ for simplicity.

Using a CBF, we can define an input set
\begin{equation}
 \set{U}_{\text{cbf}}(x) = \{u\in\set{U} \:\vert \:  \dot{h}(x,u) \geq - \gamma(h(x))\} \,
\end{equation}
that renders the system safe~\cite{ames2019a}.
This requires choosing a class-$\K_e$ function~$\gamma$ that yields $\set{U}_{\text{cbf}}(x) \neq \emptyset$ for all $x \in \set{C}$. For a controller $\pi(x)$ that is not initially designed to be safe, one can formulate a QP to modify the control input such that the system is guaranteed to be safe~\cite{ Ames2014}
\begin{subequations}
	\label{eqn:cbf_qp}
	\begin{align}
	\u_\text{s}(\x) = \underset{\u \in \set{U}}{\text{argmin}} & \quad \frac{1}{2} \lVert \u - \pi(x) \rVert_2^2 \label{eqn:cbf_qp_cost} \\ \text{s.t.} & \quad \dot{h}(x, u) \geq - \gamma(h(\x))\,. \label{eqn:cbf_constraint}
	\end{align}
\end{subequations}
Intuitively, the optimization problem in~\eqref{eqn:cbf_qp} finds an input in $\set{U}_\text{cbf}(x)$ that best matches $\pi(x)$, where the closeness of the inputs is specified with respect to a chosen distance measure~(e.g., the Euclidean norm in~\eqref{eqn:cbf_qp_cost}).
\section{Control Invariance Under Input Constraints}
\label{sec:input-constraints}
In order for a continuously differentiable function $h$ to be a valid CBF, we need to find a $\gamma\in\K_{e}$ such that~\eqref{eq:cbf_lie_derivative} holds. In general, finding such a class-$\K_e$ function under input constraints is not trivial. In~\cite{Zeng2021}, the authors proposed a method that scales a pre-selected class-$\K_e$ function online to guarantee that the safety filter optimization problem in~\eqref{eqn:cbf_qp} is feasible to solve. However, as noted by the authors, this 
does not guarantee control invariance of the safe set but rather only point-wise feasibility inside the safe set.

In this section, we assume fully known functions~$f$ and~$g$ to derive conditions for control invariance of the CBF's superlevel sets under compact input constraints $\set{U}$. Based on these conditions, we further propose a method to increase the size of feasible control input sets. We will then extend the discussion to uncertain systems in \autoref{sec:learning_prob_cont_inv_cond}.

\subsection{Valid Class-$\K_e$ Functions Under Input Constraints}
\label{subsec:valid_class_kapp_in_cont}

In this subsection, we devise a systematic approach to determine a class-$\K_e$ function for a given $h$ that satisfies \eqref{eq:cbf_lie_derivative} under input constraints $\set{U}$. 

Before we start, we first note that, in order for the problem to be well-posed, we need to check whether there exists an input~$u$ to make the Lie derivative $\dot{h}(x,u)$ non-negative for all $x$ on the boundary of the safe set~$\set{C}$:
\begin{equation}
\label{eqn:boundary_condition}
    \forall x \in \partial\set{C}, \exists u \in \set{U}, \text{~s.t.~} \dot{h}(x, u) \geq 0. 
\end{equation}
If the condition in \eqref{eqn:boundary_condition} is not satisfied for some $x\in\partial \set{C}$, then $\dot{h}$ will be negative at the particular $x$, which implies that the system will leave the safe set.
In this case, we cannot render the system safe, and the problem of finding a class-$\K_e$ function satisfying \eqref{eq:cbf_lie_derivative} will be infeasible.

In addition to the constraint on the Lie derivative at the boundary,~\eqref{eq:cbf_lie_derivative} requires the Lie derivative to be lower-bounded by a negative class-$\K_e$ function $-\gamma(h(x))$ for all $x \in \set{X}$.
The class-$\K_e$ function needs to be chosen such that there exists a $u \in \set{U}$ such that the constraint in~\eqref{eq:cbf_lie_derivative} is feasible for all $x \in \set{X}$. In order to facilitate the design of the class-$\K_e$ function, we derive a condition that ensures the feasibility of \eqref{eq:cbf_lie_derivative}.

Note that, given the compact input set~$\set{U}$, we can determine a maximum realizable Lie derivative $\dot{h}$ for each $x\in \set{X}$. We define the lower bound on the maximum realizable Lie derivative for each level set $\set{C}_c = \{ x \in \R^n~\vert ~ h(x) = c \}$ as
\begin{equation}
\label{eq:gamma-constraint}
    \gammalb (c) = - \min_{x \in \set{C}_c} \max_{u \in \set{U}} \dot{h}(x, u) \,,
\end{equation}
where $c \in \left[h_{\text{min}}, h_{\text{max}}\right]$ with $ h_{\mathrm{min}}$ and $h_{\mathrm{max}}$ corresponding to the largest and the smallest level sets of~$h$ contained in the compact set~$\set{X}$, respectively.
The following theorem then specifies a condition on $\gamma\in\K_e$ that guarantees the satisfaction of \eqref{eq:cbf_lie_derivative} for a given $h$ over the largest superlevel set of~$h$ contained in~$\set{X}$~(i.e., $\set{X}_{\text{max}} = \{ x \in \set{X}~\vert ~ h(x) \geq h_{\text{min}} \}$).

\begin{theorem}
\label{th:gamma-lower-bound}
Consider the CBF candidate $h(x)$ and the dynamics defined in~\eqref{eq:nonlinear_affine_control} with control input $u \in \set{U}$.
    Let $\gamma$ be a class-$\K_e$ function and $\gammalb(c)$ be the lower bound on the maximum realizable Lie derivative for a level set~$\set{C}_c$~(see~\eqref{eq:gamma-constraint}). 
    If $\gamma$ satisfies $\gammalb (c) \leq \gamma(c)$ for all $c\in\left[h_{\text{min}}, h_{\text{max}}\right]$, then~\eqref{eq:cbf_lie_derivative} is satisfied for all $x \in \set{X}_{\text{max}}$.     
\end{theorem}
\begin{proof}
    We first consider a level set $\set{C}_c = \{ x \in \R^n~\vert ~ h(x) = c \}$ with $c\in\left[h_{\text{min}}, h_{\text{max}}\right]$. For all $x \in \set{C}_c$, we have $\max_{\u \in \set{U}} \dot{h}(x, u) \geq \min_{x \in \set{C}_c} \max_{u \in \set{U}} \dot{h}(x, u)=- \gammalb (c) = - \gammalb (h(x))$. 
     It follows that, if $\gammalb (c) \leq \gamma(c)$,  then $ - \gammalb (h(x)) \geq - \gamma(h(x))$ and  $\max_{\u \in \set{U}} \dot{h}(x, u) \geq - \gammalb (h(x)) \geq - \gamma(h(x))$ for all $x \in \set{C}_c$. Since $h$ is continuously differentiable, $\set{X}$ is compact, and $c \in \left[h_{\text{min}}, h_{\text{max}}\right]$, the condition $\max_{\u \in \set{U}} \dot{h}(x, u) \geq - \gamma(h(x))$ holds for all~$x\in \set{X}_{\text{max}}$.
\end{proof}
 
The definition of $\gammalb$ also allows us to identify control invariant superlevel sets of $h(x)$.
They are given as follows: 
\begin{corollary}
\label{th:superlevel-set-invariance}
    Consider the CBF candidate $h(x)$, the dynamics defined in~\eqref{eq:nonlinear_affine_control} with control input $u \in \set{U}$, and $\gammalb$ as defined in~\eqref{eq:gamma-constraint}. Then, for $c \in \left[ h_{\text{min}}, h_{\text{max}} \right]$, $\gammalb (c) \geq 0$ implies control invariance of~$\Tilde{\set{C}}_c = \{x \in \set{X} \:\vert ~ h_c(x) = h(x) - c \ge 0 \}$.  
\end{corollary}
\begin{proof}
    According to~\eqref{eq:gamma-constraint}, $\gammalb (c) \geq 0$ implies that, for all states in the level set $\set{C}_c$, there exists $u \in \set{U}$ such that the Lie derivative $\dot{h}(x,u)$ is non-negative. Then, control invariance follows from standard CBF arguments~\cite{ames2019a}.
\end{proof}

Alternatively, for a known system, a subset $\hat{\set{C}} \subset \set{C}$ that is control invariant under the control input constraint set $\set{U}$ can be found using the approach proposed in~\cite{Agrawal2021}. 

In practice, to compute $\gammalb (c)$ for any $c\in [h_{\text{min}}, h_{\text{max}}]$, we can rewrite the maximization problem in~\eqref{eq:gamma-constraint} in the epigraph form~\cite{boyd2004convex}, which yields~\eqref{eq:gamma-constraint}
$ - \gammalb (c) = \min_{\rho \in \R, x \in \set{C}_c} \rho \,, \text{s.t.} \,,  \dot{h}(x, u) \leq \rho, \forall u \in \set{U} \,.$
This optimization is a nonlinear program that 
is typically hard to solve. 
It is possible to reduce its complexity by removing the level set constraint and approximating the optimization in two steps: First, determine $\rho^{j, *}_c = \min_{\rho \in \R} \rho \,, \text{s.t.} \,,  \dot{h}(x^j, u) \leq \rho, \forall u \in \set{U} \,,$ for a set of $N > 0$ samples $\{x^j \}_{j = 1}^N$ with $x^j \in \set{C}_c$.
Then approximate $-\gammalb(c) \approx \min_{j} \rho^{j, *}_c$. Choosing $N$ sufficiently large yields the desired accuracy of the approximation. 
In the special case where $\set{U}$ is a convex polytopic set, we can leverage the control-affine structure in~\eqref{eq:nonlinear_affine_control} and replace $\forall u \in \set{U}$ by $\forall u^i \in \text{vertices}(\set{U})$ (i.e., enforcing the constraint on all the vertices of~$\set{U}$). 
The above procedure applies to a single level set $\set{C}_c$. To approximate $- \gammalb (c)$ for all $c \in \left[h_{\text{min}}, h_{\text{max}} \right]$, we can repeat this procedure for a discrete set of $c^k \in \left[h_{\text{min}}, h_{\text{max}} \right]$, where $k \in \{1, \dots, K \}$. 
Choosing $K$ sufficiently large yields the desired accuracy over all level sets. %
The computational complexity of this procedure scales exponentially with the number of state and input dimensions.

\subsection{Increasing Safe Control Input Set Size}
In \autoref{subsec:valid_class_kapp_in_cont}, we presented conditions for the existence of $u \in \set{U}$ to render the system positively invariant on $\set{C}$. 
Now we turn to the problem of finding a class-$\K_e$ function that enlarges the size of the feasible control input set $\set{U}_{\text{cbf}}(x)$. This is motivated by using the CBF condition as a constraint in a minimally invasive safety filter. Ideally, the CBF safety filter should not modify a control input $u$ unless control invariance cannot be guaranteed. Formally, an uncertified control input $u$ will not be modified if $u \in \set{U}_{\text{cbf}}(x)$ for $x \in \set{X}$. We note that the set $\set{U}_{\text{cbf}}(x)$ is dependent on the choice of the class-$\K_e$ function. Therefore, it is desirable to use a class-$\K_e$ function in the CBF condition that allows the largest possible set of safe control inputs at each state~$\set{U}_{\text{cbf}}(x)$. 
An enlargement of the admissible control input set can be achieved by any class-$\K_e$ function that upper bounds $\gammalb$. Especially, any class-$\K_e$ function that approaches $+\infty$ for all $c \in \left(0, h_{\text{max}}\right]$ will satisfy the maximal feasible control input set. However, for practical purposes, it is desirable to design a class-$\K_e$ function that reflects the input constraints.
Towards this goal, we propose a target $\gamma_{\text{target}}$ for $\gamma \in \K_e$ to match in addition to satisfying the constraint in~\autoref{th:gamma-lower-bound} 
\begin{equation}
\label{eqn:gamma_target_opt}
    \gamma_{\text{target}} (c) = - \min_{x \in \set{C}_c, u \in \set{U}} \dot{h}(x, u) \,
\end{equation}
for all $c \in \left[h_{\text{min}}, h_{\text{max}}\right]$. 
The value of $-\gamma_{\text{target}}$ determines the minimum value of $\dot{h}(x, u)$ that can be achieved by
at least one state in the level set $\set{C}_c$ and one control input $u \in \set{U}$.
The constraint on $\gamma$ and the proposed target $\gamma_{\text{target}}$ can help us design a class-$\K_e$ function that optimally accounts for the input constraints. To make the design problem tractable, we consider a parameterized class-$\K_e$ function $\gamma_\theta$, where $\theta \in \Theta$ are the function's parameters.
The class-$\K_e$ function design problem is then formulated as follows:
\begin{subequations}
\label{eq:kappa_opt}
    \begin{align}
        \theta^* = \underset{\theta \in \Theta}{\text{argmin}} & \int_{h_{\text{min}}}^{h_{\text{max}}} \lVert \gamma_{\theta}(c) - \gamma_{\text{target}} (c) \rVert^2_2 \: dc \label{eq:kappa_opt_cost}\\
        \text{s.t.} & \quad \gamma_{\theta}(c) \geq \gammalb (c) \,, \forall c \in \left[h_{\text{min}}, h_{\text{max}}\right] \,. \label{eq:kappa_opt_constraint}
    \end{align}
\end{subequations}
In the case of a linear parameterization, we recover an offline version of the presented approaches in~\cite{xiao2022, Zeng2021}. However, 
our approach aims to certify a CBF candidate function~$h$ as a CBF with a corresponding class-$\K_e$ function \textit{a priori} such that online adaptation for the feasibility of the CBF condition under input constraints is not required. 
We emphasize that skipping the verification of a CBF under input constraints and solely relying on an online adaptive approach can lead to safety violations. 
This is due to~\eqref{eqn:boundary_condition} being independent of the choice of the class-$\K_e$ function. Then, adapting its steepness has no impact. Therefore, the adaptive approach can only be safely applied if the safe set~$\set{C}$ has previously been verified as control invariant under~$\set{U}$.
We demonstrate the insufficiency of the adaptive approach in a simulation example in~\autoref{sec:Simulation}.
Furthermore, the softening of the CBF condition~(as in~\cite{xiao2022, Zeng2021}) may not be desirable in practice, as
the class-$\K_e$ function also impacts the implementation of the CBF safety filter in a discrete-time setting. In discrete time, it is undesirable to approach the boundary too quickly to avoid an overshooting of the safe set boundary. Additional constraints on the slope around the origin can be incorporated for $\gamma_{\theta}$ in~\eqref{eq:kappa_opt}. 

\section{Learning Probabilistic Control Invariance Conditions under Input Constraints}
\label{sec:learning_prob_cont_inv_cond}
In this section, we address the issue of guaranteeing control invariance of the safe set $\set{C}$ under control input constraints when the dynamics are uncertain.

\subsection{Constructing Probabilistic Control Invariance Conditions Under Input Constraints}
\label{sec:probabilistic}
 We formulate a probabilistic learning approach to approximate the Lie derivative $\dot{h}(x,u)$.
 The probabilistic model could, for example, be Gaussian processes~(GPs) or ensembles of deep neural networks~\cite{DSL2021}. 

\begin{assumption}
\label{as:well-calibrated} We assume that the probabilistic model for approximating $\dot{h}(x,u)$ is well-calibrated~(i.e., has reliable uncertainty bounds)
    with $\mu(x, u) = \bar{\alpha}(x) + \bar{\beta}(x) u$ and $\sigma^2(x, u)$ denoting the posterior mean and variance functions of the probabilistic model approximating the Lie derivative $\dot{h}(x, u)$, respectively. 
    We further assume that, for any $s > 0$, there exists a  $\delta \in \left(0, 1 \right]$ such that   
    $\text{Pr}(\dot{h}(x, u) \geq \mu(x, u) - s \sigma(x, u)) \ge 1-\delta$ holds for all $(x, u) \in \set{X}\times \set{U}$.
\end{assumption}

We note that \autoref{as:well-calibrated} is not restrictive in practice. If the distribution's probability density function (PDF) is known, we can generally derive a tight confidence bound $(1-\delta)$ for a given $s>0$ by directly integrating the PDF. For a distribution with an unknown PDF but with a finite mean~$\mu$ and a finite non-zero standard deviation~$\sigma$, we can relate~$\delta$ and $s$  as  $\delta(s)=1/s^2$ (for $s>1$ cases) by applying Chebyshev's inequality~\cite{billingsley2012probability}. As an example, with Chebyshev's inequality, we can rewrite the confidence bound in \autoref{as:well-calibrated} as $\text{Pr}(\dot{h}(x, u) \geq \mu(x, u) - s \sigma(x, u)) \ge 1-1/s^2$, $\forall (x, u) \in \set{X}\times \set{U}$.
 Using~\autoref{as:well-calibrated}, we can construct different probabilistic lower bounds of $\dot{h}$ by scaling the standard deviation with an appropriate $s > 0$. In~\autoref{subsec:quadrotor_experiments}, we outline steps to verify this assumption empirically.

Following the discussion in~\autoref{sec:input-constraints}, a feasible control invariance condition with probability of at least $(1 - \delta)$ requires that, for all $x \in \partial \set{C}$, there exists an input $u \in \set{U}$ such that 
\begin{equation}
\label{eqn:prob_feasibility}
    \bar{\alpha}(x) + \bar{\beta}(x) u - s \sigma(x, u) \geq 0
\end{equation}
for some $s > 0$ (see the well-learned case in~\autoref{fig:cbf-level-set-invariance}). 
Intuitively,~$s$ represents the number of standard deviations away from the mean such that the CBF condition holds.
If the standard deviation $\sigma$ at the boundary is large, then we require a smaller number of standard deviations $s$ to ensure that there exists a control input $u$ satisfying~\eqref{eqn:prob_feasibility} 
with probability~$(1 - \delta)$ at any time instance. If no $s >0$ satisfies~\eqref{eqn:prob_feasibility} or the largest~$s$ satisfying \eqref{eqn:prob_feasibility} yields an undesirably low probability~$(1-\delta)$, then additional data needs to be collected to reduce the predictive uncertainty of the learned Lie derivative model.%
\begin{figure}[t]
    \centering
    \includegraphics[width=\columnwidth, height=5.0cm]{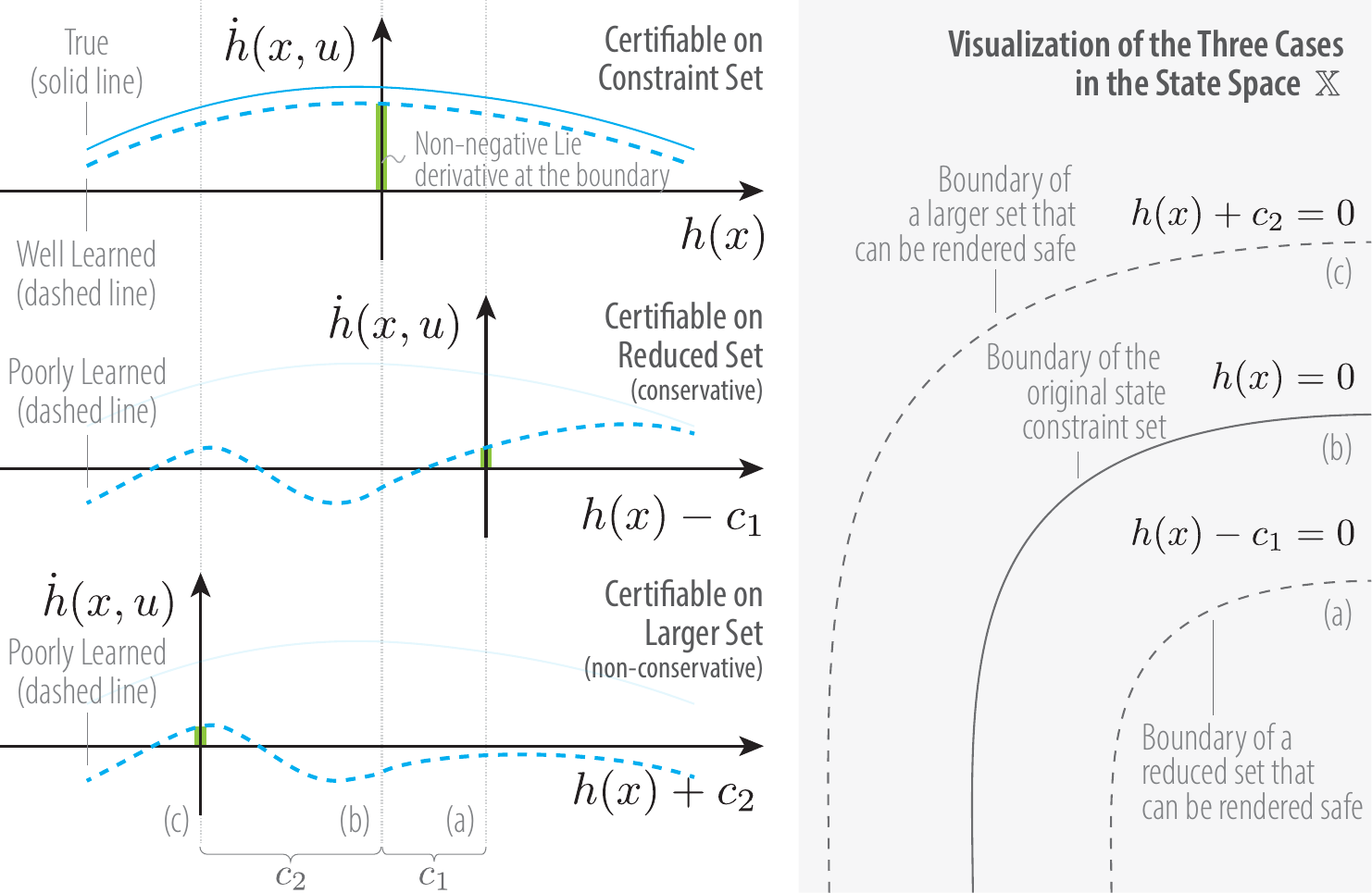}
    \vspace*{-5mm}
    \caption{Illustration showing which superlevel sets of $h(x)$ can be rendered control invariant. The solid blue line represents the true Lie derivative~$\dot{h}(x, u)$, and the dashed blue line represents a learned lower bound $\mu(x, u) - s \sigma(x, u)$. If there exists at least one control input such that a non-negative Lie derivative lower bound can be achieved for all states in a level set $\set{C}_c$, then we can guarantee control invariance with high probability. 
    }
    \label{fig:cbf-level-set-invariance}
\end{figure}
Once the desired predictive accuracy is achieved, the next step is determining a class-$\K_e$ function $\hat{\gamma}$ for the estimated Lie derivative under input constraints. As the true Lie derivative is unknown, we use probabilistic lower bounds of the Lie derivative to provide safety guarantees. We can leverage the approach for finding class-$\K_e$ functions under input constraints presented in~\autoref{sec:input-constraints}. Substituting $\mu(x, u) - s \sigma(x, u)$ for $\dot{h}(x, u)$ when finding $\gammalb$ and $\gamma_{\text{target}}$ leads to an optimization problem as formulated in~\eqref{eq:kappa_opt} for $\hat{\gamma}$.
Given the mean~$\mu$ and the standard deviation~$\sigma$, 
we provide a sufficient condition for satisfying the CBF condition for the true system on~$\set{C}$ with high probability.
\begin{theorem}
\label{th:prob-cbf}
Consider the CBF candidate $h(x)$, the input constraint set~$\set{U}$, and a chosen $s>0$.
Under~\autoref{as:well-calibrated}, if the estimated lower bound of the true Lie derivative $\mu(x, u) - s \sigma(x, u)$ is non-negative for $x \in \partial \set{C}$, then there exists $\hat{\gamma} \in \K_e$ such that 
\begin{equation}
    \label{eq:prob-cbf-constraint}
        \max_{u \in \set{U}} \left[ \mu(x, u) - s \sigma(x, u) \right] \geq - \hat{\gamma}(h(x))\,, \forall x \in \set{C} \,
    \end{equation}
holds. Moreover, we can guarantee that the CBF condition $\dot{h}(x, u) \geq -\hat{\gamma}(h(x) )$ is satisfied by~\eqref{eq:nonlinear_affine_control} for some input $u\in\set{U}$ at any time instance with a probability of at least~$(1 - \delta)$. 
\end{theorem}
\begin{proof}
    Let $\hat{\gammalb}(c) = - \min_{x \in \set{C}_c} \max_{u \in \set{U}} \mu(x, u) - s \sigma(x, u)$ similar to~\eqref{eq:gamma-constraint}. 
    Due to the non-negativity of $\mu(x, u) - s \sigma(x, u)$ for all $x \in \partial \set{C}$, we can find a $\hat{\gamma} \in \K_e$ such that $\hat{\gammalb}(c) \leq \hat{\gamma}(c)$ for all $c\in \left[ 0, h_{\text{max}} \right]$. Following a similar argument as in the proof for~\autoref{th:gamma-lower-bound}, one can show that~\eqref{eq:prob-cbf-constraint} is feasible for all $x \in \set{C}$. Furthermore, by~\autoref{as:well-calibrated}, $\text{Pr}(\dot{h}(x, u) \geq \mu(x, u) - s \sigma(x, u)) \ge 1-\delta$. For all $x \in \set{C}$, the CBF condition $\dot{h}(x, u) \geq -\hat{\gamma}(h(x) )$ can be then satisfied by the true system  for some $u\in\set{U}$ with a probability of at least $(1 - \delta)$. 
\end{proof}

We can use this result to formulate a safety filtering optimization for the uncertain system similar to that defined in~\eqref{eqn:cbf_qp}. 
For this, we substitute the true Lie derivative $\dot{h}(x, u)$ in constraint~\eqref{eqn:cbf_constraint} with the estimated lower bound on the Lie derivative $\mu(x, u) - s \sigma(x, u)$, which yields
\begin{subequations}
    \label{eqn:cbf_variance}
	\label{eqn:cbf_qp_variance}
	\begin{align}
	\u_\text{s}(\x) = \underset{\u \in \set{U}}{\text{argmin}} & \quad \frac{1}{2} \lVert \u - \pi(x) \rVert_2^2 \\ \text{s.t.} & \quad \mu(x, u) - s \sigma(x, u) \geq - \hat{\gamma}(h(x))\,. \label{eqn:cbf_constraint_variance}
	\end{align}
\end{subequations}

Extension to control invariance guarantees over a finite time interval can be considered in future work~\cite{lew2021problem}. 

\begin{remark}
    In the case where the non-negativity of $\mu(x, u) - s \sigma(x, u)$  is not satisfied for all $x \in \partial \set{C}$, one may potentially render a subset of $\set{C}$ safe if there exists a $c \in \left(0,h_{\text{max}}\right]$ such that $\hat{\gammalb}(c) \ge 0 $ holds (see the second case in \autoref{fig:cbf-level-set-invariance} and \autoref{th:superlevel-set-invariance}). If such $c\in \left(0,h_{\text{max}}\right]$ does not exist, then one may only be able to render a larger set safe (see the third case in \autoref{fig:cbf-level-set-invariance}); in this case, we are no longer able to provide closed-loop safety guarantees, and the probabilistic model needs to be retrained using a larger dataset.
\end{remark}

\subsection{Learning Probabilistic Lie Derivative Models}
\label{sec:learning}
Suppose we are given sampled trajectory data as a training dataset. 
Then, we can determine a probabilistic model of $\dot{h}$ based on data using standard supervised learning methodology. 
The inputs to the probabilistic model are the state~$x$ and the control input~$u$, and the output of the model is an estimate of the Lie derivative $\dot{\hat{h}}(x, u)$.  
In practice, such a dataset could be collected through expert demonstrations. 

In the following, we consider an ensemble of $M > 1$ function approximators as our probabilistic model. For example, this could be an ensemble of deep neural networks with different weight initializations trained on the same training data set $\mathcal{D}_{\text{train}}$. 
We can retain the control affine structure using an ensemble of deep neural networks.
For details on retaining the control affine structure with GPs, see~\cite{Castaneda2021acc}.
We choose each function approximator as $\dot{\hat{h}}_i(x, u) = \alpha_i(x) + \beta_i(x) u$ with $\alpha_i(x)$ and $\beta_i(x)$ learned jointly using supervised learning. 
Then, the mean function is $\mu(x, u) = \bar{\alpha}(x) + \bar{\beta}(x) u = \frac{1}{M} \sum_{i = 1}^M \dot{\hat{h}}_i(x, u)$ and the variance function is $\sigma^2(x, u) = \frac{1}{M - 1} \sum_{i = 1}^M (\dot{\hat{h}}_i(x, u) - \mu(x, u))^2$. 

The variance function is no longer affine in $u$ and solving~\eqref{eqn:cbf_qp_variance} generally requires nonlinear programming. 
In the special case where the set $\set{U}$ is convex and polytopic, we can reformulate the constrained optimization problem in~\eqref{eqn:cbf_qp_variance} as a second-order cone program~(SOCP) with a linear objective and two second-order cone constraints as in~\cite{Castaneda2021acc}. 
We have $\sigma(x, u) = \lVert L(x) u + l(x) \rVert_2$, where $L(x)$ is the decomposition of a positive semidefinite matrix 
$S(x)= L^\intercal(x) L(x) = \frac{1}{M-1}\sum_{i = 1}^M(\beta_i(x) - \bar{\beta}(x))^\intercal(\beta_i(x) - \bar{\beta}(x)) \succeq 0$ and $l(x) = \sqrt{\frac{1}{M-1}\sum_{i = 1}^M(\alpha_i(x) - \bar{\alpha}(x))^2}$. Plugging this into~\eqref{eqn:cbf_constraint_variance} yields a second-order cone constraint. The quadratic objective in~\eqref{eqn:cbf_qp_variance} can be replaced by the second-order cone constraint $\frac{1}{\sqrt{2}}\lVert u - \pi(x) \rVert_2 \leq j$ and by minimizing the linear objective $j$. 

\begin{remark}
    For uncertain systems, the Lie derivative needs to be estimated numerically~\cite{taylor2020a}. The resulting estimation error can be accounted for in the probabilistic model~\cite{l4dc22}. 
\end{remark}

\section{Evaluation Results}
\subsection{Simulation Results}
\label{sec:Simulation}
We evaluated the feasibility under input constraints of our proposed offline approach compared to the online adaptive approach presented in~\cite{xiao2022, Zeng2021}. We consider the damped inverted pendulum with the nonlinear control affine dynamics
    $\dot{x}_1 = x_2 \,, ~ \dot{x}_2 = 1 / (m l^2) \left(- m g l \sin (x_1) - b x_2 + u \right)\,,$
with input constraints $\set{U} = \{u \in \R \vert -\SI{15}{\newton \meter} \leq u \leq \SI{15}{\newton \meter} \}$, mass~$m = \SI{1}{\kilogram}$, length~$l = \SI{0.5}{\meter}$, damping factor~$b = \SI{0.1}{\frac{\newton}{\meter}}$, and the gravitational constant~$g = \SI{9.81}{\frac{\meter}{\second^2}}$. We consider a CBF candidate $h(x)$ that is a valid CBF for the pendulum system without input constraints. 
We assume full model knowledge and the uncertified control policy is $\pi(x) = 0$.

Our approach assesses feasibility under input constraints offline before employing the safety filter online. 
Based on our proposed approach and the remarks on the implementation for convex polytopic input sets~$\set{U}$ in~\autoref{sec:input-constraints}, we verify that the superlevel set of $h(x) = c = 0.005$ as $\Tilde{\set{C}}_{c} = \{ x \in \set{X} \vert h(x) - c \geq 0 \}$~(dark gray in~\autoref{fig:pendulum}) can be rendered control invariant under the constraint~$\set{U}$.
In contrast, the online adaptive approach employs a strategy that relies on adapting the steepness of the lower bound in the CBF condition online. 
However, using this strategy for a candidate safe set $\set{C}$~(light gray in~\autoref{fig:pendulum}) that has not been verified \textit{a priori} does not guarantee control invariance.
On the boundary of an uncertifiable safe set, the lower bound on the CBF condition cannot be adapted and is zero by definition. In such a case, an online adaptive approach is insufficient for safety.
In simulation, we find that for initial conditions starting in our verified control invariant safe set~$\Tilde{\set{C}}_{0.005}$, the adaptive approach may leave this verified safe set and can even leave the larger uncertified candidate safe set~$\set{C}$, see~\autoref{fig:pendulum}.

\begin{figure}[tb]
    \centering
 \input{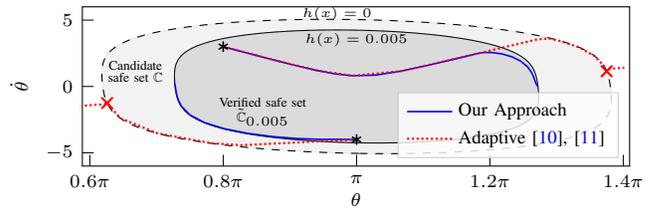}
 \vspace*{-5mm}
    \caption{Simulation of an inverted pendulum in closed-loop with safety filters using either our proposed offline input constraint verification (solid blue trajectories) and the online adaptive modification~\cite{xiao2022, Zeng2021} (dotted red trajectories) given a CBF candidate~$h$. 
    We use our approach to verify the control invariance of a superlevel set $\Tilde{\set{C}}_{0.005}$ (dark gray) of the CBF candidate~$h$ (zero superlevel set $\set{C}$ in light gray).
    Starting from the same initial conditions (indicated by stars) inside the verified safe set, our approach guarantees control invariance, while the online adaptive approach cannot guarantee safety and results in the system leaving the candidate safe set $\set{C}$ as indicated by the red crosses.}
    \label{fig:pendulum}
\end{figure}

\subsection{Quadrotor Experiments}
\label{subsec:quadrotor_experiments}
We validated our proposed method through physical experiments on a miniature quadrotor, the Crazyflie. 2.1.
In our experiments, our proposed safety filter for an uncertain input-constrained system prevents the falling quadrotor from crashing into the floor. 
In our experiment, we constrain the motion of the quadrotor to its $z$-axis by stabilizing its attitude. This gives a system representation with state $x = \begin{bmatrix} z_{\text{pos}} & z_{\text{vel}} \end{bmatrix}^\intercal \in \mathbb{X} \subseteq \mathbb{R}^2$ and the collective delta thrust $\Delta u \in \mathbb{U} \subseteq \mathbb{R}$ as the input. The delta thrust is defined as  $\Delta u = u - u_0 \in \left[-0.086mg, +0.086mg \right]$
where $u_0$ is the collective thrust to achieve hover, $m = \SI{0.03}{\kilogram}$ is the quadrotor's mass, and $g = 9.8~\text{m}/ \text{s}^{2}$ is the gravitational constant. 
The quadrotor position $z_{\text{pos}}$ and velocity $z_{\text{vel}}$ are estimated by a Luenberger observer using position measurements from a motion capture system.
We collect state and input data by sending sinusoidal inputs with constant amplitudes and frequencies. The input constraint set $\set{U}$ is the closed interval of the minimum and maximum applied inputs during the data collection. We use a nominal double integrator linear time-invariant~(LTI) model and, from the collected data, fit an ensemble of five neural networks~(NN) to the Lie derivative residual to augment the safety filter. Given a CBF candidate $h(x) = 1 - (x - x_{\text{c}})^{\intercal} P (x - x_{\text{c}})$ with $P$ being symmetric and positive definite and $x_c$ the center of the safe set, we apply our approach from~\autoref{sec:learning_prob_cont_inv_cond} to verify the zero-superlevel set~$\set{C}$ of the CBF candidate using the learned Lie derivative residual with $s=3$ under input constraints~$\set{U}$. Using Chebyshev’s inequality, this yields $\delta(s) = \frac{1}{s^2} = \frac{1}{9}$. During the experiments, we start from hovering with an approximate initial condition $x_{0} = [1.4, 0.0]^{\intercal}$. Then, an uncertified control policy $\pi(x) = -0.03mg$ is applied with the SOCP-based safety filter~\eqref{eqn:cbf_variance}. 
After each execution, we verified the satisfaction of~\autoref{as:well-calibrated} for our proposed approach by determining the percentage of predicted Lie derivative residuals inside the confidence interval given by $s = 3$.
The experiment results are presented in~~\autoref{fig:state_trajecotry} and~\autoref{fig:quad-h-dot}. A video of this experiment can be found here: \href{http://tiny.cc/CBF}{\texttt{http://tiny.cc/CBF}}.

The state trajectories without safety filter and with safety filter for the LTI system and the learned Lie derivative are shown in~\autoref{fig:state_trajecotry}. Without the safety filter, the quadrotor leaves the safe set and hits the floor. With the safety filter that leverages the learned Lie derivative and the nominal LTI system, the system strictly stays inside the safe set, while the safety filter solely relying on the LTI system fails. The learned Lie derivative using the ensemble of neural networks allows us to explicitly account for uncertainty in the system, including model mismatch and disturbances in real-world systems. The effect of different linear class-$\K_e$ functions can be seen in~\autoref{fig:quad-h-dot}. 
Class-$\K_e$ functions define the tolerable Lie derivative~$\dot{h}$ to guarantee safety when the system approaches the safe set's boundary.
The quadrotor approaches the boundary of the safe set (and therefore the ground) faster for less conservative class-$\K_e$ like $\gamma_3$. A less steep function~(e.g., as given by~$\gamma_1$) achieves a more conservative behavior. In this case, the optimization of the class-$\K_e$ is not restrictive, as the maximum realizable Lie derivative is positive for all level sets $h(x) \in \left[0.0, 0.7\right]$. However, the class-$\K_e$ must be carefully designed to achieve safe operation at discrete time steps.

\begin{figure}[tb]
    \centering
    \input{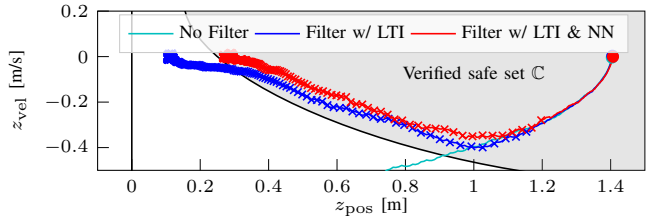}
    \vspace*{-2mm}
    \caption{Quadrotor state trajectories under unsafe dropping input (cyan) and CBF-augmented input using an LTI system (blue) and an LTI system with learned Lie derivative residual (red). The safe set with its interior~(light grey) and its boundary (black) is defined as an ellipse. Solid circles indicate the initial states. Crosses indicate the states where the safety filter is active. }
    \label{fig:state_trajecotry}
    \vspace*{-2mm}
\end{figure}

\begin{figure}[tb]
    \centering
    \input{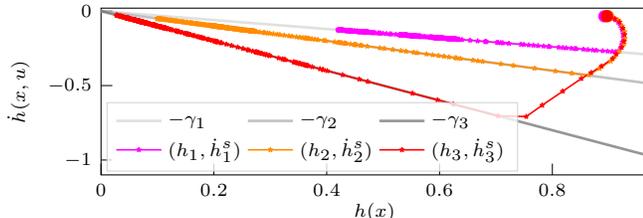}
    \vspace*{-2mm}
    \caption{Real-world trajectories of $h$ and $\dot{h}^s = \dot{h}_{\text{nominal}}(x, u) + \mu(x, u) - s \sigma(x, u)$, where $\dot{h}_{\text{nominal}}$ is the nominal Lie derivative and $s = 3$, for three different linear extended class-$\K_e$ functions~(in different shades of gray). Solid circles indicate the initial states.}
    \label{fig:quad-h-dot}
\end{figure}

\section{Conclusion}
In this paper, we proposed a systematic approach for learning valid CBF conditions for uncertain systems subject to input constraints. 
Based on known dynamics, we first derived a sufficient condition on the class-$\K_e$ function for validating a CBF safety filter design under input constraints. We proposed an approach to determine optimal class-$\K_e$ functions for a given CBF. Building on these results, we developed a learning-based method to construct probabilistic control invariance conditions for input-constrained uncertain systems. Given the probabilistic CBF condition, we can design a CBF safety filter to guarantee the system's safety with high probability. 
In simulation and real-world robotic experiments, we demonstrated the efficacy of the learned safety filters and showed that optimized class-$\K_e$ functions can be less conservative and are significant for designing safe real-world operation.

\balance

\bibliographystyle{./IEEEtranBST/IEEEtran}
\bibliography{./IEEEtranBST/IEEEabrv,./references}

\begin{thebibliography}{10}
\providecommand{\url}[1]{#1}
\csname url@rmstyle\endcsname
\providecommand{\newblock}{\relax}
\providecommand{\bibinfo}[2]{#2}
\providecommand\BIBentrySTDinterwordspacing{\spaceskip=0pt\relax}
\providecommand\BIBentryALTinterwordstretchfactor{4}
\providecommand\BIBentryALTinterwordspacing{\spaceskip=\fontdimen2\font plus
\BIBentryALTinterwordstretchfactor\fontdimen3\font minus \fontdimen4\font\relax}
\providecommand\BIBforeignlanguage[2]{{%
\expandafter\ifx\csname l@#1\endcsname\relax
\typeout{** WARNING: IEEEtran.bst: No hyphenation pattern has been}%
\typeout{** loaded for the language `#1'. Using the pattern for}%
\typeout{** the default language instead.}%
\else
\language=\csname l@#1\endcsname
\fi
#2}}

\bibitem{DSL2021}
L.~Brunke, M.~Greeff, A.~W. Hall, Z.~Yuan, S.~Zhou, J.~Panerati, and A.~P. Schoellig, ``Safe learning in robotics: From learning-based control to safe reinforcement learning,'' \emph{Annual Review of Control, Robotics, and Autonomous Systems}, vol.~5, pp. 411--444, 2022.

\bibitem{Ames2014}
A.~D. Ames, J.~W. Grizzle, and P.~Tabuada, ``{Control barrier function based quadratic programs with application to adaptive cruise control},'' \emph{Proc. of the IEEE Conf. on Decision and Control (CDC)}, pp. 6271--6278, 2014.

\bibitem{Dai2022}
H.~Dai and F.~Permenter, ``Convex synthesis and verification of control-{L}yapunov and barrier functions with input constraints,'' in \emph{Proc. of the IEEE American Control Conf. (ACC)}, 2023, accepted.

\bibitem{ames2019a}
A.~D. Ames, S.~Coogan, M.~Egerstedt, G.~Notomista, K.~Sreenath, and P.~Tabuada, ``Control barrier functions: Theory and applications,'' in \emph{Proc. of the European Control Conf. (ECC)}, 2019, pp. 3420--3431.

\bibitem{Nagumo1942berDL}
M.~Nagumo, ``{{\"U}ber die Lage der Integralkurven gew{\"o}hnlicher Differentialgleichungen},'' \emph{Proc. of the Physico-Mathematical Society of Japan, 3rd Series}, vol.~24, pp. 551--559, 1942.

\bibitem{mitchell2005time}
I.~Mitchell, A.~Bayen, and C.~Tomlin, ``A time-dependent {H}amilton-{J}acobi formulation of reachable sets for continuous dynamic games,'' \emph{IEEE Trans. on Automatic Control}, vol. 50(7), pp. 947--957, 2005.

\bibitem{wang2023a}
H.~Wang, K.~Margellos, and A.~Papachristodoulou, ``Safety verification and controller synthesis for systems with input constraints,'' \emph{IFAC-PapersOnLine}, 2023, 22nd IFAC World Congress, accepted.

\bibitem{Wei2023}
T.~Wei, S.~Kang, W.~Zhao, and C.~Liu, ``Persistently feasible robust safe control by safety index synthesis and convex semi-infinite programming,'' \emph{IEEE Control Systems Letters}, vol.~7, pp. 1213--1218, 2023.

\bibitem{Qin2022}
Z.~Qin, D.~Sun, and C.~Fan, ``Sablas: Learning safe control for black-box dynamical systems,'' \emph{IEEE Robotics and Automation Letters}, vol.~7, no.~2, pp. 1928--1935, 2022.

\bibitem{xiao2022}
W.~Xiao, C.~Belta, and C.~G. Cassandras, ``Adaptive control barrier functions,'' \emph{IEEE Transactions on Automatic Control}, vol.~67, no.~5, pp. 2267--2281, 2022.

\bibitem{Zeng2021}
J.~Zeng, B.~Zhang, Z.~Li, and K.~Sreenath, ``Safety-critical control using optimal-decay control barrier function with guaranteed point-wise feasibility,'' in \emph{Proc. of the IEEE American Control Conf. (ACC)}, 2021, pp. 3856--3863.

\bibitem{Agrawal2021}
D.~R. Agrawal and D.~Panagou, ``Safe control synthesis via input constrained control barrier functions,'' in \emph{Proc. of the IEEE Conf. on Decision and Control (CDC)}, 2021, pp. 6113--6118.

\bibitem{taylor2020a}
A.~Taylor, A.~Singletary, Y.~Yue, and A.~Ames, ``Learning for safety-critical control with control barrier functions,'' in \emph{Proc. of the Learning for Dynamics and Control Conf. (L4DC)}, vol. 120, 2020, pp. 708--717.

\bibitem{l4dc22}
L.~Brunke, S.~Zhou, and A.~P. Schoellig, ``Barrier {Bayesian} linear regression: {Online} learning of control barrier conditions for safety-critical control of uncertain systems,'' in \emph{Proc. of the Learning for Dynamics and Control Conf. (L4DC)}, vol. 168, 2022, pp. 881--892.

\bibitem{Wang2018a}
L.~Wang, E.~A. Theodorou, and M.~Egerstedt, ``Safe learning of quadrotor dynamics using barrier certificates,'' in \emph{Proc. of the IEEE Intl. Conf. on Robotics and Automation (ICRA)}, 2018, pp. 2460--2465.

\bibitem{Ohnishi2019}
M.~Ohnishi, L.~Wang, G.~Notomista, and M.~Egerstedt, ``Barrier-certified adaptive reinforcement learning with applications to brushbot navigation,'' \emph{IEEE Trans. on Rob.}, vol. 35(5), pp. 1186--1205, 2019.

\bibitem{Castaneda2021}
F.~Castañeda, J.~J. Choi, B.~Zhang, C.~J. Tomlin, and K.~Sreenath, ``Pointwise feasibility of gaussian process-based safety-critical control under model uncertainty,'' in \emph{Proc. of the IEEE Conf. on Decision and Control (CDC)}, 2021, pp. 6762--6769.

\bibitem{Nejati2023}
A.~Nejati and M.~Zamani, ``Data-driven synthesis of safety controllers via multiple control barrier certificates,'' \emph{IEEE Control Systems Letters}, vol.~7, pp. 2497--2502, 2023.

\bibitem{boyd2004convex}
S.~Boyd, S.~P. Boyd, and L.~Vandenberghe, \emph{Convex {O}ptimization}.\hskip 1em plus 0.5em minus 0.4em\relax Cambridge University Press, 2004.

\bibitem{billingsley2012probability}
P.~Billingsley, \emph{Probability and Measure}, ser. Wiley Series in Probability and Statistics.\hskip 1em plus 0.5em minus 0.4em\relax Wiley, 2012.

\bibitem{lew2021problem}
T.~Lew, A.~Sharma, J.~Harrison, E.~Schmerling, and M.~Pavone, ``On the problem of reformulating systems with uncertain dynamics as a stochastic differential equation,'' 2021.

\bibitem{Castaneda2021acc}
F.~Castañeda, J.~J. Choi, B.~Zhang, C.~J. Tomlin, and K.~Sreenath, ``Gaussian process-based min-norm stabilizing controller for control-affine systems with uncertain input effects and dynamics,'' in \emph{2021 American Control Conference (ACC)}, 2021, pp. 3683--3690.

\end{thebibliography}

\balance
\end{document}